\documentclass[12pt]{iopart}

\usepackage{graphicx}
\usepackage{rotating}

\usepackage{iopams}

\newtheorem{theorem}{Theorem}

\newtheorem{corollary}{Corollary}

\newtheorem{lemma}{Lemma}

\newtheorem{proposition}{Proposition}

\newenvironment{proof}[1][Proof]{\textbf{#1.} }{\ \rule{0.5em}{0.5em}}

\usepackage{diagrams}
\newarrow{XK}{|}{-}{-}{-}{triangle}

\begin{document}

\title[]{Non-solvable contractions of semisimple Lie algebras in low dimension }

\author{R. Campoamor-Stursberg\dag}

\address{\dag\ Dpto. Geometr\'{\i}a y Topolog\'{\i}a\\Fac. CC. Matem\'aticas\\
Universidad Complutense de Madrid\\Plaza de Ciencias, 3\\E-28040
Madrid, Spain}

\ead{rutwig@mat.ucm.es}

\begin{abstract}
The problem of non-solvable contractions of Lie algebras is
analyzed. By means of a stability theorem, the problem is shown to
be deeply related to the embeddings among semisimple Lie algebras
and the resulting branching rules for representations. With this
procedure, we determine all deformations of indecomposable Lie
algebras having a nontrivial Levi decomposition onto semisimple
Lie algebras of dimension $n\leq 8$, and obtain the non-solvable
contractions of the latter class of algebras.

\end{abstract}


\maketitle


\newpage

\section{Introduction}

Contractions of Lie algebras have played a major role in physical
applications, starting from the pioneering work of Segal and
In\"on\"u and Wigner \cite{Seg,IW} up to the many generalizations
of the contraction notion developed over the decades \cite{Sa}.
Early in the development of the theory of contractions, its
relation to a somewhat inverse procedure, that of deformations of
Lie algebras, was recognized and developed in \cite{LN}, and
tested for consistency in the case of three dimensional algebras.
An important consequence of this work was the fact that the Lie
algebras contracting onto a given Lie algebra $\frak{g}$ had to be
searched among the deformations of the latter, thus establishing
the invertibility of contractions.\footnote{Later it was pointed
out that not every deformation is associated to a contraction.}
The introduction of further techniques like the cohomology of Lie
algebras \cite{Ri} allowed one to interpret contractions
geometrically in the variety of Lie algebras having a fixed
dimension. Once the most important groups intervening in
applications were analyzed, like the Lie algebras in the classical
and quantum relativistic kinematics, the attention of various
authors was turned to obtain complete diagrams of contractions in
low dimension \cite{Con}, which have been enlarged and completed
in order to cover all the special types of contractions considered
earlier \cite{Po}. Such lists have been obtained up to dimension 4
over the field of real numbers. In this approach, the analysis
depends essentially on a reliable classification of real Lie
algebras, which only exists up to dimension six. For higher
dimensions, only partial results have been obtained, and the
absence of a classification of solvable non-nilpotent algebras
constitutes an important obstruction in studying contractions for
any fixed dimension.

\medskip

In this work, we approach the contraction problem from another
point of view. Instead of fixing the dimension, we focus on the
structure of the contracting Lie algebras. To this extent, we
choose the semisimple Lie algebras up to dimension 8, and
determine the non-solvable contractions. It turns out that the
Levi decomposition and the embedding problem of semisimple Lie
algebras, as well as the branching rules of representations, play
a prominent role in this analysis. Actually, Levi subalgebras of
Lie algebras have a certain stability property that allows one to
control, up to some extent, how the deformations and contractions
behave \cite{Ri}. Using the reversibility of contractions, we
determine the deformations of low dimensional Lie algebras
$\frak{g}$ having a nontrivial Levi decomposition, i.e., such that
they decompose into $\frak{g}\overrightarrow{\oplus}_{R}\frak{r}$
with $\frak{s}\neq 0$ semisimple, $\frak{r}\neq 0$ the radical and
$R$ a nontrivial representation of the semisimple part acting by
derivations on the radical. In particular, we determine which
deformations lead to a semisimple Lie algebra, and obtain the
corresponding contraction. For decomposable contractions, i.e.,
algebras decomposing as direct sum of ideals, we find that they
exist whenever none of the ideals is semisimple. This will imply
that reductive algebras can only appear as contractions of
decomposable algebras.
\medskip

Unless otherwise stated, any Lie algebra $\frak{g}$ considered in
this work is defined over the field $\mathbb{R}$ of real numbers.
We convene that nonwritten brackets are either zero or obtained by
antisymmetry. We also use the Einstein summation convention.
Abelian Lie algebras of dimension $n$ will be denoted by the
symbol $nL_{1}$.

\section{Contractions, deformations and cohomology of Lie
algebras}

From the geometrical point of view, a Lie algebra
$\frak{g}=(V,\mu)$ is a pair formed by a vector space $V$ and a
bilinear alternating (i.e., skew-symmetric) tensor $\mu:V\times
V\rightarrow V$ that satisfies the Jacobi identity. For any fixed
basis of $V$, the coordinates of this tensor are identified with
the structure constants $C_{ij}^{k}$ of $\frak{g}$. In this sense,
the set of real Lie algebra laws $\mu$ over $V$ forms a variety
$\mathcal{L}^{n}$ embedded in $\mathbb{R}^{\frac{n^3-n^2}{2}}$
\cite{Ni}. The coordinates of a point correspond to the structure
tensor of an algebra $\frak{g}$. Since the general linear group
acts naturally on this variety, the orbits $\mathcal{O}(\frak{g})$
of a point $\frak{g}$ (i.e., a Lie algebra) are formed by all Lie
algebras isomorphic to $\frak{g}$. Deformations of Lie algebras
arise from the problem of studying the properties of these orbits.
This leads one to analyze neighborhoods of a given Lie algebra in
the variety, as well as the intersection of orbits corresponding
to different Lie algebras. Of special interest are the so called
stable Lie algebras, which are those for which the orbit
$\mathcal{O}(\frak{g})$ is an open set \cite{Ri}. One of the main
tools in this analysis is the adjoint cohomology of Lie algebras
\cite{Ni}.

Recall that an n-cochain $\varphi$ of a Lie algebra
$\frak{g}=(V,\mu=[. , .])$ is a multilinear antisymmetric map
$\varphi:V\times.^{n}.\times V\rightarrow M$, where $M$ is a
$\frak{g}$-module. For the special case $M=V$, we get the vector
space $C^{n}(V,V)$ of $n$-cochains with values in the adjoint
module.\footnote{By the identification of $\frak{g}$ with the pair
$(V,\mu)$, we can further suppose that the Lie bracket $[.,.]$ is
given by $[X,Y]=\mu(X,Y)$ for all $X,Y\in V$.} By means of the
coboundary operator
\begin{eqnarray}
d\varphi(X_{1},..,X_{n+1})=\sum_{i=1}^{n+1}(-1)^{i+1}\left[
X_{i},\varphi(X_{1},..,\widehat{X}_{i},..,X_{n+1})\right]+\nonumber \\
\sum_{1\leq i,j\leq
n+1}(-1)^{i+j}\varphi\left(\left[X_{i},X_{j}\right],X_{1},..,\widehat{X}_{i},..,\widehat{X}_{j},..X_{n+1}\right)
\label{KRO}
\end{eqnarray}
we obtain a cochain complex $\left\{d:C^{n}(V,V)\rightarrow
C^{n+1}(V,V),\quad n\geq 0\right\}$. In particular, $d\circ d=0$
holds. We call $\varphi\in C^{n}(V,V)$ a $n$-cocycle if
$d\varphi=0$, and a $n$-coboundary if there exists $\sigma\in
C^{n-1}(V,V)$ such that $d\sigma=\varphi$. The spaces of cocycles
and coboundaries are denoted by $Z^{n}(V,V)$, respectively
$B^{n}(V,V)$. By (\ref{KRO}), we have the inclusion relation
$B^{n}(V,V)\subset Z^{n}(V,V)$ for all $n$, and the quotient space
\begin{equation}
H^{n}(V,V)=Z^{n}(V,V)/ B^{n}(V,V)
\end{equation}
is called $n$-cohomology space of $\frak{g}$ for the adjoint
representation \cite{Se}. Among the many applications of these
spaces \cite{Ri,Vi,Az2,Ch}, they are relevant for the study of
orbits in the following sense. A formal one-parameter deformation
$\frak{g}_{t}$ of a Lie algebra $\frak{g}=(V,\left[.,.\right])$ is
given by a deformed commutator:
\begin{equation}
\left[X,Y\right]_{t}:=\left[X,Y\right]+\psi_{m}(X,Y)t^{m},\label{DK}
\end{equation}
where $t$ is a parameter and $\psi_{m}:V\times V\rightarrow V$ is
a skew-symmetric bilinear map. Imposing that these formal brackets
satisfy the Jacobi identity (up to quadratic order of $t$), one
obtains the following expression:
\begin{eqnarray}
\fl
\left[X_{i},\left[X_{j},X_{k}\right]_{t}\right]_{t}+\left[X_{k},\left[X_{i},X_{j}\right]_{t}\right]_{t}+
\left[X_{j},\left[X_{k},X_{i}\right]_{t}\right]_{t}\nonumber \\
\lo = t d\psi_{1}(X_{i},X_{j},X_{k})+
t^{2}\left(\frac{1}{2}\left[\psi_{1},\psi_{1}\right]+
d\psi_{2}\right)(X_{i},X_{j},X_{k})+ \mathcal{O}(t^{3}),
\label{JG}
\end{eqnarray}
where $d\psi_{l}$ is the trilinear map of (\ref{KRO}) for $n=2$
and $\left[\psi_{1},\psi_{1}\right]$ is defined by
\begin{equation}
\fl
\frac{1}{2}\left[\psi_{1},\psi_{1}\right](X_{i},X_{j},X_{k}):=\psi_{1}\left(\psi_{1}(X_{i},X_{j}),X_{k}\right)+
\psi_{1}\left(\psi_{1}(X_{j},X_{k}),X_{i}\right)+\psi_{1}\left(\psi_{1}(X_{k},X_{i}),X_{j}\right).
\end{equation}
If equation (\ref{JG}) equals zero, then we have the conditions
\begin{eqnarray}
d\psi_{1}(X_{i},X_{j},X_{k})=0, \label{K2}\\
\frac{1}{2}\left[\psi_{1},\psi_{1}\right](X_{i},X_{j},X_{k})+d\psi_{2}(X_{i},X_{j},X_{k})=0.
\label{K3}
\end{eqnarray}
Equation (\ref{K2}) shows that $\psi_{1}$ is a 2-cocycle in
$H^{2}(\frak{g},\frak{g})$, implying that deformations are
generated by 2-cocycles\footnote{By this we mean that the linear
term of the deformation is a cocycle.}. On the other hand,
equation (\ref{K3}) implies that the deformation satisfies a
so-called integrability condition. Additional conditions are
obtained if the deformed bracket is developed up to higher orders
of $t$ \cite{Ni,Az2}. In particular, if for some $\psi_{1}\in
Z^{2}(\frak{g},\frak{g})$ we have
$\left[\psi_{1},\psi_{1}\right]=0$, then the cocycle is called
integrable and the linear deformation $\frak{g}+t \psi_{1}$
defines a Lie algebra.

If the algebra $\frak{g}_{t}$ is isomorphic to $\frak{g}$, the
deformation $\frak{g}_{t}$ is called trivial. It is not difficult
to show that if this happens, then we can find a non-singular map
$f_{t}:V\rightarrow V$ such that
$f_{t}\left(\left[X,Y\right]_{t}\right)=\left[f_{t}X,f_{t}Y\right]$
for all $X,Y\in V$. This means that $\psi_{1}=df_{t}$, and the
cocycle is trivial (i.e., a coboundary). Therefore trivial
deformations are generated by 2-coboundaries \cite{Vi}. In this
framework, contractions and deformations of Lie algebras can be
related using trivial deformations \cite{Ch}.

\medskip

Classically, a contraction is defined as follows: Let $\frak{g}$
be  a Lie algebra and $\Phi_{t}\in Aut(\frak{g})$ a family of
non-singular linear maps of $\frak{g}$, where $t\in [1,\infty)$.
For any $X,Y\in\frak{g}$, the bracket over the transformed basis
has the form
\begin{equation}
\left[X,Y\right]_{\Phi_{t}}:=\Phi_{t}^{-1}\left[\Phi_{t}(X),\Phi_{t}(Y)\right].
\end{equation}
If the limit
\begin{equation}
\left[X,Y\right]_{\infty}:=\lim_{t\rightarrow
\infty}\Phi_{t}^{-1}\left[\Phi_{t}(X),\Phi_{t}(Y)\right]
\label{Ko}
\end{equation}
exists for any $X,Y\in\frak{g}$, then equation (\ref{Ko}) defines
a Lie algebra $\frak{g}^{\prime}$ called the contraction of
$\frak{g}$ (by $\Phi_{t}$), non-trivial if $\frak{g}$ and
$\frak{g}^{\prime}$ are non-isomorphic Lie algebras. Further, it
is not difficult to see that the infinitesimal version of equation
(\ref{Ko}) is generated by a coboundary \cite{Vi}. In fact, if we
consider a trivial cocycle $\psi\in B^{2}(\frak{g},\frak{g})$, let
$\sigma$ be the 1-cochain such that $d\sigma=\psi$. Using the
exponential map we obtain the linear transformation
$f_{t}=\exp(-t\sigma)$, and expressing the brackets over the
transformed basis $\left\{f_{t}(X_{i}\right\}$, we get
\begin{equation}
\left[X,Y\right]_{t}=f_{t}^{-1}\left[f_{t}(X),f_{t}(Y)\right].
\label{Ko1}
\end{equation}
Therefore a contraction can be obtained by taking limits in
(\ref{Ko1}). An important result states that for any contraction
of Lie algebras $\frak{g}\rightarrow \frak{g}^{\prime}$ there is a
deformation of $\frak{g}^{\prime}$ that reverses it \cite{LN}.
However, it should be remarked that a formal deformation is not
necessarily related to a contraction \cite{We,C63}.

A special case is given when $H^{2}(\frak{g},\frak{g})=0$. In this
situation, the Lie algebra $\frak{g}$ has no nontrivial
deformations, and, in particular, cannot arise as a contraction.
Such algebras are therefore stable. Although stable algebras with
nonvanishing cohomology exist \cite{Ni}, this condition implies
the stability of important classes of Lie algebras, such as
semisimple and parabolic Lie algebras \cite{To}.

\medskip

\section{Contractions and cohomology}

By the preceding results, contractions of Lie algebras can be
analyzed using cohomological tools. More specifically, the
deformations of Lie algebras are computed, and those being
invertible provide contractions \cite{Ch,C24}.

In general, the effective computation of the cohomology of Lie
algebras is a difficult task. However, for the case of Lie
algebras having a non-trivial Levi decomposition, there exists a
useful reduction, called the Hochschild-Serre spectral sequence
\cite{Se}. If $\frak{g}$ has the Levi decomposition
$\frak{g}=\frak{s}\overrightarrow{\oplus}_{R}\frak{r}$, where
$\frak{s}$ denotes the Levi subalgebra, $\frak{r}$ the radical of
$\frak{g}$ and $R$ a representation of $\frak{s}$ that acts by
derivations on the radical \cite{Tu}, then the adjoint cohomology
$H^{p}(\frak{g},\frak{g})$ admits the following decomposition:
\begin{equation}
H^{p}\left(  \frak{g},\frak{g}\right)  \simeq\sum_{i+j=p}H^{i}\left(  \frak{g}%
,\mathbb{R}\right)  \otimes H^{j}\left(  \frak{r},\frak{g}\right)
^{\frak{g}},
\end{equation}
where $H^{j}\left(  \frak{r},\frak{g}\right)  ^{\frak{g}}$ is the
space of $\frak{g}$-invariant cocycles. These are the multilinear
skew-symmetric maps $\varphi\in C^{j}(\frak{r},\frak{s})$ that
satisfy the coboundary operator (\ref{KRO}) and such that
\begin{eqnarray*}
\fl
(X\varphi)(Y_{1},..,Y_{j})=\left[X,\varphi(Y_{1},..,Y_{j})\right]
-\sum_{i=1}^{j}\varphi\left(\left[X,Y_{i}\right],Y_{1},,.\widehat{Y}_{i},..,Y_{j}\right)=0,\\
\lo \forall X\in\frak{s},\; Y_{1},..,Y_{j}\in \frak{r}.
\end{eqnarray*}
For the particular case $p=2$ the formula simplifies to
\begin{equation}
H^{2}\left(  \frak{g},\frak{g}\right)  \simeq H^{2}\left(
\frak{r},\frak{g}\right) ^{\frak{s}}. \label{HS}
\end{equation}
This result suggests that Levi subalgebras are stable in some
sense, and that deformations are determined by appropriate
modification of the brackets in the radical. This idea actually
constitutes an important theorem that will be used later.

\begin{proposition}
Let $\frak{g}=\frak{s\oplus r}$ be the direct sum of a semisimple
Lie algebra $\frak{s}$ and an arbitrary algebra $\frak{r}$. Then
$H^{2}\left( \frak{g},\frak{g}\right)  \simeq H^{2}\left(
\frak{r},\frak{r}\right)  $.
\end{proposition}

\begin{proof}
By the Hochschild-Serre spectral sequence, formula (\ref{HS})
holds. As an $\frak{r}$-module, the space $H^{2}\left(
\frak{r},\frak{g}\right) ^{\frak{g}}$ is trivial \cite{Se}, and
this implies that
\begin{equation*}
H^{2}\left(  \frak{r},\frak{g}\right)  ^{\frak{s}}\simeq
H^{2}\left( \frak{r},\frak{g}\right)  ^{\frak{g}}.
\end{equation*}
It suffices therefore to consider the $\frak{s}$-invariance. Now,
for any $\varphi\in$ $H^{2}\left(  \frak{r},\frak{g}\right)
^{\frak{s}}$ and $X\in\frak{s},\,Y,Z\in\frak{r}$ we have
\begin{equation}
\fl \left(  X\varphi\right)  \left(  Y,Z\right)  =\left[
X,\varphi\left( Y,Z\right)  \right]  -\varphi\left(  \left[
X,Y\right]  ,Z\right) -\varphi\left(  Y,\left[  X,Z\right] \right)
=\left[  X,\varphi\left( Y,Z\right)  \right]  =0.\label{IB1}
\end{equation}
because the sum is direct. Now $\varphi\left(  Y,Z\right)
\in\frak{g}$, and by the decomposition of $\frak{g}$ we can
rewrite it as
\begin{equation}
\varphi\left(  Y,Z\right)
=W_{1}+W_{2},\;W_{1}\in\frak{s},\,W_{2}\in \frak{r}.
\end{equation}
Since $\frak{s}$ is semisimple, for any $X$  there exists $X^{\prime}%
\in\frak{s}$ such that $\left[  X,X^{\prime}\right]  \neq0$. By
the invariance condition (\ref{IB1}) we must have $W_{1}=0$, thus
$\varphi\left(  Y,Z\right)  \in \frak{r}$ for all
$Y,Z\in\frak{r}$. This proves that any invariant cochain is
actually a 2-cochain of the radical, from which the assertion
follows by imposing the coboundary condition.
\end{proof}

\begin{corollary}
Let $\frak{g}$ be an indecomposable Lie algebra with non-trivial
Levi subalgebra $\frak{s}$. Then $\frak{g}$ cannot contract onto a
direct sum
$\frak{s}^{\prime}\oplus\frak{r}$ of a semisimple Lie algebra $\frak{s}%
^{\prime}$ with an arbitrary Lie algebra $\frak{r}$.
\end{corollary}

A direct consequence of this property is that deformations of
reductive Lie algebras $\frak{s}\oplus n L_{1}$ are always
decomposable. In particular, they cannot appear as contractions of
indecomposable Lie algebras having a nontrivial Levi decomposition
or semisimple Lie algebras. Moreover, any Lie algebra
$\frak{s}\oplus\frak{t}$ with $\frak{t}$ an arbitrary
$n$-dimensional algebra, contracts onto the reductive algebra
$\frak{s}\oplus n L_{1}$. This result does obviously not exclude
the possibility that an indecomposable Lie algebra contracts onto
a non-solvable decomposable algebra, it merely states that none of
the ideals intervening in the decomposition can be semisimple.
Large classes of Lie algebras having this type of contractions
exist, like semidirect products of semisimple and Heisenberg Lie
algebras \cite{C46}.

\section{Contractions of semisimple Lie algebras}

The previous interpretation of Lie algebras as points of a variety
provides us with some useful criteria to study deformations and
contractions. In \cite{Ri}, an important result concerning the
topology of orbits was obtained.  It makes precise the intuitive
idea about stability of Levi subalgebras observed previously.

\begin{theorem}{\rm \cite{Ri}}
Let $L=(V,\mu)$ be a Lie algebra, $\frak{s}$ a semisimple
subalgebra of $L$ and $\frak{r}$ the complementary subspace of
$\frak{s}$ in $V$. There exists a neighborhood $U^{\mu}\in
\mathcal{L}^{n}$ of $\mu$ such that if $\mu_{1}\in U^{\mu}$, then
the algebra $L_{1}=(V,\mu_{1})$ is isomorphic to a Lie algebra
$L^{\prime}=(V,\mu^{\prime})$ that satisfies the conditions

\begin{enumerate}
\item $\mu(X,X^{\prime})=\mu^{\prime}(X,X^{\prime}), \forall
X,X^{\prime}\in\frak{s}$, \item $\mu(X,Y)=\mu^{\prime}(X,Y),
\forall X\in\frak{s},Y\in \frak{r}$.
\end{enumerate}
\end{theorem}

In essence, this stability theorem, due to Page and Richardson
\cite{Ri}, establishes that if the Lie algebra $\frak{g}$ has a
semisimple subalgebra $\frak{s}$, then its deformations will have
some subalgebra isomorphic to $\frak{s}$, and that the action of
$\frak{s}$ on the remaining generators is preserved. Combined with
the Hochschild-Serre spectral sequence, this result tells that the
main information about deformations of semidirect products is
codified in the radical of the algebra. As application of this
theorem, we can establish the following result for non-solvable
contractions of semisimple Lie algebras:

\begin{proposition}
Let $\frak{g}=\frak{s}\overrightarrow{\oplus}_{R}\frak{r}$ be a
contraction of a semisimple Lie algebra $\frak{s}^{\prime}$. Then
the following holds:
\begin{enumerate}

\item there exists some semisimple subalgebra $\frak{s}_{1}$ of
$\frak{s}^{\prime}$ isomorphic to $\frak{s}$,

\item identifying $\frak{s}$ with $\frak{s}_{1}$ via an
isomorphism, the adjoint representation of $\frak{s}^{\prime}$
decomposes as $ad(\frak{s}^{\prime})|_{\frak{s}}=
ad(\frak{s})\oplus R$ with respect to the embedding
$\frak{s}\hookrightarrow \frak{s}^{\prime}$.

\item $\frak{g}$ has at least ${\rm rank}(\frak{s}^{\prime})$
independent Casimir operators.
\end{enumerate}
\end{proposition}

The proof is nothing but a slight variation of the stability
theorem. If $\frak{g}=\frak{s}\overrightarrow{\oplus}_{R}\frak{r}$
is a contraction of $\frak{s}^{\prime}$, then there exists some
deformation of $\frak{g}$ reversing the contraction \cite{We}. By
the stability theorem, this deformation has some subalgebra that
is isomorphic to the Levi part of $\frak{g}$, and acts the same
way on the generators of the radical. Therefore the embedding of
semisimple Lie algebras $\frak{s}\hookrightarrow
\frak{s}^{\prime}$ induces a branching rule for representations,
and the quotient algebra $\frak{s}^{\prime}/\frak{s}$, seen as an
$\frak{s}$-module, is isomorphic to the representation $R$, that
is, $ad(\frak{s}^{\prime})|_{\frak{s}}= ad(\frak{s})\oplus R$.
This proves (i) and (ii). Finally, the third condition follows
from the properties of contractions of invariants \cite{C24}.

\begin{corollary}
Let $\frak{s}$ be a semisimple Lie algebra of a semisimple algebra
$\frak{s}^{\prime}$, and $R$ be a representation of $\frak{s}$. If
$ad(\frak{s}^{\prime})|_{\frak{s}}\neq ad(\frak{s})\oplus R$, then
no Lie algebra with Levi decomposition
$\frak{s}\overrightarrow{\oplus}_{R}\frak{r}$ ( $\frak{r}$
solvable) can arise as a contraction of $\frak{s}^{\prime}$.
\end{corollary}

The problem of analyzing the non-solvable contractions of
semisimple Lie algebras $\frak{s}^{\prime}$ is therefore reduced
to analyze the deformations of Lie algebras having Levi
decomposition $\frak{s}\overrightarrow{\oplus}_{R}\frak{r}$, where
$\frak{s}$ is some semisimple subalgebra of $\frak{s}^{\prime}$,
$R$ is obtained from the branching rules with respect to the
embedding $\frak{s}\hookrightarrow \frak{s}^{\prime}$ and
$\frak{r}$ is a solvable Lie algebra. In view of the
Hochschild-Serre reduction theorem, whether such a deformation
onto a semisimple algebra is possible or not depends essentially
on the structure of the radical $\frak{r}$. In general, the
following cases can appear when studying the deformations
$\frak{g}_{t}$ of $\frak{s}\overrightarrow{\oplus}_{R}\frak{r}$:
\begin{enumerate}

\item $\frak{s}$ is a maximal semisimple subalgebra of
$\frak{s}^{\prime}$, and either $\frak{g}_{t}$ is isomorphic to
$\frak{s}^{\prime}$ or there exists a solvable Lie algebra
$\frak{r}^{\prime}$ such that $\frak{g}_{t}\simeq
\frak{s}\overrightarrow{\oplus}_{R}\frak{r}^{\prime}$.

\item $\frak{s}$ is not a maximal semisimple subalgebra of
$\frak{s}^{\prime}$. In this case, a deformation $\frak{g}_{t}$
that is not semisimple is either isomorphic to a semidirect
product $\frak{s}\overrightarrow{\oplus}_{R}\frak{r}^{\prime}$
with $\frak{r}^{\prime}$ solvable, or there exists a semisimple
subalgebra $\frak{s}_{1}$ of $\frak{s}^{\prime}$ and a
representation $R_{1}$ of $\frak{s}_{1}$ such that
$\frak{g}_{t}\simeq
\frak{s}_{1}\overrightarrow{\oplus}_{R_{1}}\frak{r}^{\prime}$ for
some solvable Lie algebra $\frak{r}^{\prime}$. If the latter
holds, then we have the chain $\frak{s}\hookrightarrow
\frak{s}_{1}\hookrightarrow \frak{s}^{\prime}$ of semisimple Lie
algebras, and the branching rule $ad(\frak{s}_{1})\oplus R_{1} =
ad(\frak{s})\oplus R$ is satisfied.
\end{enumerate}

Case (ii) is typical when we consider contractions of simple Lie
algebras onto double inhomogeneous Lie algebras \cite{He,He1,C49}.

\section{Contractions of semisimple Lie algebras in low dimension}

In order to determine all nonsolvable contractions of semisimple
Lie algebras, we must have a classification of indecomposable Lie
algebras having a nontrivial Levi decomposition. The first such
classifications in dimensions $n\leq 8$ are due to Turkowski
\cite{Tu}. In the following, we use the notation of this paper to
label the Lie algebras. For completeness, the structure constants
are given in Table A1 of the appendix. For notational purposes, we
adopt the convention that the term $D_{j}$ for
$j\in\frac{1}{2}\mathbb{Z}$ denotes the irreducible representation
with highest weight $2j$ of $\frak{sl}(2,\mathbb{R})$, while
$R_{4}^{II}$ and $R_{5}$ are, respectively, the irreducible
representations of dimension 4 and 5 of $\frak{so}(3)$.

\medskip

Since the Levi decomposition of Lie algebras is trivial up to
dimension four, we will obtain nontrivial results from dimension 5
onwards. In the following, we suppose that
$\frak{g}=\frak{s}\overrightarrow{\oplus}_{R}\frak{r}$ is an
indecomposable Lie algebra with nontrivial Levi decomposition.

\subsection{$\dim\frak{g}=5$}

If $\dim\frak{g}=5$, then $\frak{g}$ must be isomorphic to the
special affine Lie algebra
$\frak{sa}(2,\mathbb{R})=\frak{sl}(2,\mathbb{R})
\overrightarrow{\oplus}_{D_{\frac{1}{2}}}2L_{1}$. Since it is the
only Lie algebra in this dimension having a nontrivial Levi
decomposition, and no five dimensional semisimple algebras exists,
it must be stable, and does not arise as a contraction. Further,
by corollary 1, any contraction of $\frak{sa}(2,\mathbb{R})$ is
necessarily a solvable Lie algebra.

\subsection{$\dim\frak{g}=6$}

In dimension six, we have the simple Lorentz algebra
$\frak{so}(3,1)$ and the semisimple Lie algebras
$\frak{so}(4)=\frak{so}(3)\oplus\frak{so}(3)$,
$\frak{so}^{*}(4)=\frak{so}(3)\oplus\frak{sl}(2,\mathbb{R})$ and
$\frak{so}(2,2)=\frak{sl}(2,\mathbb{R})\oplus\frak{sl}(2,\mathbb{R})$.
Additionally, four indecomposable non-solvable algebras $L_{6,j}$
($j=1..4$) with nontrivial Levi decomposition exist (see
appendix). By proposition 2, only the algebras
$L_{6,1}=\frak{so}(3)\overrightarrow{\oplus}_{ad}3L_{1}$ and
$L_{6,4}=\frak{sl}(2,\mathbb{R})\overrightarrow{\oplus}_{D_{1}}3L_{1}$
can arise as a contraction of a semisimple algebra. Since these
algebras are isomorphic to the inhomogeneous algebras
$I\frak{so}(3)$ and $I\frak{sl}(2,\mathbb{R})$, we recover the
well known contractions \cite{We1}:
\begin{equation}
\begin{diagram}
\frak{so}\left(  2,2\right)   &  &  &  & \frak{so}\left(  4\right)   & \\
\dTo &   & \frak{so}\left(  3,1\right)   &  & \dTo & \\
& \ldTo  &  &  \rdTo &  & \\
I\frak{sl}\left(  2,\mathbb{R}\right)   &  &  &  &
I\frak{so}\left(  3\right) &
\end{diagram}
\end{equation}
On the contrary, the semisimple algebra $\frak{so}^{*}(4)$ has no
indecomposable contractions with nonzero Levi part. Any
non-solvable contraction of it is necessarily a direct sum of a
simple and a solvable algebra. For the remaining algebras,
$L_{6,2}$ and $L_{6,3}$, it is not difficult to show, applying the
Hochschild-Serre reduction theorem, that the identities
$H^{2}(L_{6,2},L_{6,2})=H^{2}(L_{6,3},L_{6,3})=0$ hold, from which
we deduce that these algebras are stable. It is trivial to verify
that both $L_{6,2}$ and $L_{6,3}$ contract onto the decomposable
algebra $\frak{sa}(2,\mathbb{R})\oplus L_{1}$. We resume the
situation in the following

\begin{proposition}
A six dimensional indecomposable Lie algebra with nontrivial Levi
decomposition is either stable or the contraction of a semisimple
Lie algebra.
\end{proposition}

We remark that this result, in combination with corollary 1,
provide a complete analysis of the non-solvable contractions of
semisimple Lie algebras in this dimension.

\subsection{$\dim\frak{g}=7$}

Although there are no semisimple Lie algebras in dimension seven,
this dimension is of interest, since we find the lowest
dimensional examples of non-solvable Lie algebras that do not
arise as a contraction, but are nevertheless not stable. According
to the classification in \cite{Tu}, there are seven isomorphism
classes $L_{7,j}$, one of them depending on a continuous parameter
$p$.

\begin{proposition}
Let $\frak{g}=L_{7,j}$ with $j\neq 3$. Then
$H^{2}(L_{7,j},L_{7,j})=0$ and $L_{7,j}$ is stable.\newline If
$\frak{g}=L_{7,3}^{p}$, then
\begin{equation*}
\dim H^{2}\left(
L_{7,3}^{p},L_{7,3}^{p}\right)  =\left\{
\begin{tabular}
[c]{ll}%
$1,$ & $p\neq2$\\
$2,$ & $p=2$%
\end{tabular}
\right.  .
\end{equation*}
Moreover, the cohomology classes are given by $\left[
\varphi_{1}\right]  $ and $\left[  \varphi_{2}\right]  $, where
\[
\varphi_{1}\left(  X_{6},X_{7}\right)  =X_{6};\;\varphi_{2}\left(  X_{4}%
,X_{5}\right)  =X_{6}.
\]
\end{proposition}
The proof follows by direct computation using the Hochschild-Serre
reduction. In particular, all deformations of $L_{7,3}^{p}$ for
$p\neq 2$ lie in the same family, and no contraction among these
algebras is possible since $\dim Der(L_{7,3}^{p})=8$ for all $p$,
and any contraction increases the number of derivations. For
$p=2$, two independent deformations are possible, since the
integrability condition (\ref{K3}) implies that
$\epsilon_{1}\epsilon_{2}=0$. The deformation
$L_{7,3}^{2}+\epsilon_{1}\varphi_{1}$ is isomorphic to
$L_{7,3}^{2+\epsilon_{1}}$, and clearly non-invertible by the
dimension of the algebra of derivations, while the deformation
$L_{7,3}^{2}+\epsilon_{2}\varphi_{2}$ has a non-abelian radical
and is easily seen to be isomorphic to $L_{7,4}$. It is
straightforward to verify that we obtain the contraction
$L_{7,4}\longrightarrow L_{7,3}^{2}$.

\begin{proposition}
A seven dimensional indecomposable Lie algebra $\frak{g} $ with
nontrivial Levi decomposition is a contraction of a Lie algebra if
and only if $\frak{g}\simeq L_{7,3}^{2}$.
\end{proposition}

In particular, the algebras of the family $L_{7,3}^{p}$ with
$p\neq 2$ are neither stable nor contractions of another algebra.
From this dimension onwards, this pathology appears in any
dimension. We observe further that, to some extent, the stability
of the remaining algebras (those not depending on a parameter) is
due to the nonexistence of semisimple Lie algebras in this
dimension.

\subsection{$\dim\frak{g}=8$}

In dimension 8, the only real semisimple Lie algebras are the real
forms of $A_{2}$, that is, the compact algebra $\frak{su}(3)$ and
the non-compact algebras $\frak{su}(2,1)$ and
$\frak{sl}(3,\mathbb{R})$. In order to obtain the contractions of
these algebras that are indecomposable and have a nonzero Levi
part, we must determine all possible embeddings of rank one simple
subalgebras and their corresponding branching rules.

\begin{proposition}
Let $\frak{s}^{\prime}\hookrightarrow \frak{s}$ be a semisimple
subalgebra of a semisimple Lie algebra $\frak{s}$ of dimension
$8$. If the indecomposable Lie algebra
$\frak{g}=\frak{s}^{\prime}\overrightarrow{\oplus}_{R}\frak{r}$ is
a contraction of $\frak{s}$, then $\frak{s}^{\prime}$ is a maximal
simple subalgebra of $\frak{s}$ and one of the following cases
holds:
\begin{enumerate}

\item $\frak{s}^{\prime}\simeq \frak{sl}(2,\mathbb{R})$ and
$R=2D_{\frac{1}{2}}\oplus D_{0}$ or $D_{1}$.

\item $\frak{s}^{\prime}\simeq \frak{so}(3)$ and $R=R_{5}$ or
$R_{4}^{II}\oplus D_{0}$.
\end{enumerate}
\end{proposition}

\begin{proof}
Since any semisimple Lie algebra in dimension 8 has rank two, a
simple subalgebra is necessarily maximal as semisimple algebra. In
order to obtain the admissible representations $R$, it suffices to
consider the complexification
$\frak{g}\otimes_{\mathbb{R}}\mathbb{C}$ of $\frak{g}$. Then the
Levi subalgebra is isomorphic to $A_{1}$, and the problem reduces
to determine the branching rule for the adjoint representation
$\Gamma(1,1)_{\mathbb{C}}$ of $A_{2}$ with respect to the
non-equivalent embeddings of $A_{1}$ in $A_{2}$. If the embedding
$A_{1}\hookrightarrow A_{2}$ is regular \cite{On} , then the
representation $\Gamma(1,1)_{\mathbb{C}}$ decomposes as
\begin{equation}
\Gamma(1,1)_{\mathbb{C}}|_{A_{1}}= D_{1}\oplus
2D_{\frac{1}{2}}\oplus D_{0}.\label{VR1}
\end{equation}
Since $D_{1}$ is the adjoint representation of $A_{1}$, the only
possibility for $R_{\mathbb{C}}$ is
$R_{\mathbb{C}}=2D_{\frac{1}{2}}\oplus D_{0}$. Taking the real
forms of $\frak{sl}(2,\mathbb{C})$, we obtain that
$R=2D_{\frac{1}{2}}\oplus D_{0}$ if $\frak{s}^{\prime}\simeq
\frak{sl}(2,\mathbb{R})$ and $R=R_{4}^{II}\oplus D_{0}$ if
$\frak{s}^{\prime}\simeq \frak{so}(3)$.\newline For the nonregular
embedding $A_{1}\hookrightarrow A_{2}$, the corresponding
branching rule is easily obtained from (\ref{VR1}), and equals
\begin{equation}
\Gamma(1,1)_{\mathbb{C}}|_{A_{1}}= ad(A_{1})\oplus
R_{\mathbb{C}}=D_{2}\oplus D_{1}.\label{VR2}
\end{equation}
Taking the real forms, we obtain the irreducible representations
$D_{2}$ if $\frak{s}^{\prime}\simeq \frak{sl}(2,\mathbb{R})$ and
$R_{5}$ if $\frak{s}^{\prime}$ is compact.
\end{proof}

In view of the classification, the only Lie algebras having the
previous describing representations are
$\frak{F}=\left\{L_{8,2},L_{8,3},L_{8,4}^{p},L_{8,5},L_{8,13}^{\epsilon},L_{8,14},L_{8,15},L_{8,16},
,L_{8,17}^{p},L_{8,18}^{p},L_{8,21}\right\}$.

\begin{lemma}
The Lie algebras $L_{8,3},L_{8,4}^{p\neq
0},L_{8,16},L_{8,17}^{p\neq -1},L_{8,18}^{p\neq 0}$ do not arise
as a contraction of a semisimple Lie algebra.
\end{lemma}

The proof follows at once observing that these algebras satisfy
the condition $\mathcal{N}(\frak{g})=0$, where
$\mathcal{N}(\frak{g})$ denotes the number of independent
invariants for the coadjoint representation. By proposition 2,
they cannot be contractions of a Lie algebra having invariants.
For the remaining algebras, the existence or not of contractions
cannot be deduced from the usual numerical invariants that are
preserved or increased by contraction.\footnote{For a list of such
invariants, see e.g. \cite{Po,C24}.} In order to analyze whether
they are contractions of semisimple Lie algebras, we determine if
they admit deformations onto semisimple algebras. To this extent,
we apply the Hochschild-Serre reduction to compute a basis of
$H^{2}(\frak{g},\frak{g})$ and analyze the deformed bracket
(\ref{DK}). The bases for the adjoint cohomology are given in
Table 1.

\begin{table}
\begin{indented}\item[]
\caption{Adjoint cohomology of Lie algebras in $\frak{F}$.}
\begin{tabular}{@{}lcl}
$\frak{g}$ & $\dim H^{2} $ & Cocycle basis of $H^{2}(\frak{g},\frak{g})$\\
\hline\mr
$L_{8,2}$ & $1$ & $%
\begin{array}
[c]{l}%
\varphi\left(  X_{4},X_{5}\right)  =X_{2},\;\varphi\left(
X_{4},X_{6}\right) =X_{3},\;\varphi\left(  X_{4},X_{7}\right)
=X_{1},\;\varphi\left(
X_{4},X_{8}\right)  =-\frac{3}{2}X_{5},\\
\varphi\left(  X_{5},X_{6}\right)  =X_{1},\,\varphi\left(
X_{5},X_{7}\right) =-X_{3},\;\varphi\left(  X_{5},X_{8}\right)
=\frac{3}{2}X_{4},\;\varphi
\left(  X_{6},X_{7}\right)  =X_{2},\\
\varphi\left(  X_{6},X_{8}\right)
=\frac{3}{2}X_{7},\;\varphi\left( X_{7},X_{8}\right)
=-\frac{3}{2}X_{6}.
\end{array}
\ \ \ $  \\\ms
$L_{8,4}^{0}$ & $2$ & $%
\begin{array}
[c]{l}%
\varphi_{1}\left(  X_{4},X_{5}\right)  =X_{2},\;\varphi_{1}\left(  X_{4}%
,X_{6}\right)  =X_{3}+\frac{3}{2}X_{8},\;\varphi_{1}\left(  X_{4}%
,X_{7}\right)  =X_{1},\;\\
\varphi_{1}\left(  X_{5},X_{6}\right)  =X_{1},\;\varphi_{1}\left(  X_{5}%
,X_{7}\right)  =-X_{3}+\frac{3}{2}X_{8},\;\varphi_{1}\left(  X_{6}%
,X_{7}\right)  =X_{2}.\\
\varphi_{2}\left(  X_{4},X_{8}\right)  =X_{4},\;\varphi_{2}\left(  X_{5}%
,X_{8}\right)  =X_{5},\;\varphi_{2}\left(  X_{6},X_{8}\right)  =X_{6}%
,\;\varphi_{2}\left(  X_{7},X_{8}\right)  =X_{7}.
\end{array}
\ $\\\ms
$L_{8,5}$ & $1$ & $%
\begin{array}
[c]{l}%
\varphi\left(  X_{4},X_{5}\right)  =X_{3},\;\varphi\left(
X_{4},X_{6}\right) =X_{2},\;\varphi\left(  X_{4},X_{7}\right)
=X_{1},\;\varphi\left(
X_{5},X_{6}\right)  =-X_{1},\\
\varphi\left(  X_{5},X_{7}\right)  =X_{2},\;\varphi\left(
X_{6},X_{7}\right) =2X_{3},\;\varphi\left(  X_{6},X_{8}\right)
=-6X_{1},\;\varphi\left( X_{7},X_{8}\right)  =-6X_{2}.
\end{array}
$\\\ms
$L_{8,13}^{\varepsilon}$ & $1$ & $%
\begin{array}
[c]{l}%
\varphi\left(  X_{4},X_{6}\right)  =-2X_{2},\;\varphi\left(  X_{4}%
,X_{7}\right)  =X_{1},\;\varphi\left(  X_{4},X_{8}\right)
=-3\varepsilon
X_{6},\;\varphi\left(  X_{5},X_{6}\right)  =X_{1},\\
\varphi\left(  X_{5},X_{7}\right)  =2X_{3},\;\varphi\left(  X_{5}%
,X_{8}\right)  =-3\varepsilon X_{7},\;\varphi\left(
X_{6},X_{8}\right) =3X_{4},\;\varphi\left(  X_{7},X_{8}\right)
=3X_{5}.
\end{array}
$\\\ms
$L_{8,14}$ & $3$ & $%
\begin{array}
[c]{l}%
\varphi_{1}\left(  X_{4},X_{8}\right)  =X_{4},\;\varphi_{1}\left(  X_{5}%
,X_{8}\right)  =X_{5};\;\varphi_{2}\left(  X_{4},X_{8}\right)  =X_{6}%
,\;\varphi_{2}\left(  X_{5},X_{8}\right)  =X_{7};\\
\varphi_{3}\left(  X_{6},X_{7}\right)  =X_{8}.
\end{array}
$\\\ms
$L_{8,15}$ & $1$ & $%
\begin{array}
[c]{l}%
\varphi\left(  X_{4},X_{5}\right)  =3X_{8},\;\varphi\left(  X_{4}%
,X_{6}\right)  =-2X_{2},\;\varphi\left(  X_{4},X_{7}\right)  =X_{1}%
,\;\varphi\left(  X_{4},X_{8}\right)  =-3X_{6},\\
\varphi\left(  X_{5},X_{6}\right)  =X_{1},\;\varphi\left(
X_{5},X_{7}\right) =2X_{3},\;\varphi\left(  X_{5},X_{8}\right)
=-3X_{7}.
\end{array}
$\\\ms
$L_{8,17}^{-1}$ & $2$ & $%
\begin{array}
[c]{l}%
\varphi_{1}\left(  X_{4},X_{6}\right) =-2X_{2},\;\varphi_{1}\left(
X_{4},X_{7}\right) =X_{1}-3X_{8},\;\varphi_{1}\left(
X_{5},X_{6}\right)
=X_{1}+3X_{8},\\
\varphi_{1}\left(  X_{5},X_{7}\right)  =2X_{3};\;\varphi_{2}\left(
X_{6},X_{8}\right)  =X_{6},\;\varphi_{2}\left(  X_{7},X_{8}\right)
=X_{7}.
\end{array}
$\\\ms
$L_{8,18}^{0}$ & $2$ & $%
\begin{array}
[c]{l}%
\varphi_{1}\left(  X_{4},X_{5}\right)  =3X_{8},\;\varphi_{1}\left(
X_{4},X_{6}\right)  =-2X_{2},\;\varphi_{1}\left(
X_{4},X_{7}\right)
=X_{1},\;\varphi_{1}\left(  X_{5},X_{6}\right)  =X_{1},\\
\varphi_{1}\left(  X_{5},X_{7}\right)  =2X_{3},\;\varphi_{1}\left(
X_{6},X_{7}\right)  =3X_{8}.\\
\varphi_{2}\left(  X_{6},X_{8}\right)  =X_{6},\;\varphi_{2}\left(  X_{7}%
,X_{8}\right)  =X_{7}.
\end{array}
$\\\ms
$L_{8,21}$ & $1$ & $%
\begin{array}
[c]{l}%
\varphi\left(  X_{4},X_{7}\right)  =-2X_{2},\;\varphi\left(  X_{4}%
,X_{8}\right)  =X_{1},\;\varphi\left(  X_{5},X_{6}\right)  =6X_{2}%
,\;\varphi\left(  X_{5},X_{7}\right)  =-2X_{1},\\
\varphi\left(  X_{5},X_{8}\right)  =2X_{3},\;\varphi\left(  X_{6}%
,X_{7}\right)  =-6X_{3}.
\end{array}
$\\ \hline\br
\end{tabular}
\end{indented}
\end{table}

\begin{theorem}
The indecomposable Lie algebras $L_{8,2},L_{8,4}^{0},L_{8,5},
L_{8,13}^{\epsilon},L_{8,14},L_{8,15}, L_{8,17}^{-1},L_{8,18}^{0}$
and $L_{8,21}$ are all obtained as contractions of simple Lie
algebras. More precisely,
\begin{enumerate}

\item $\frak{su}(3)$ contracts onto the algebras $L_{8,2}$,
$L_{8,4}^{0}$ and $L_{8,5}$.

\item $\frak{su}(2,1)$ contracts onto the algebras
$L_{8,2},L_{8,4}^{0}, L_{8,13}^{1}, L_{8,14},L_{8,15},
L_{8,18}^{0}$ and $L_{8,21}$.

\item $\frak{sl}(3,\mathbb{R})$ contracts onto the algebras
$L_{8,5}, L_{8,13}^{-1}, L_{8,14},L_{8,15}, L_{8,17}^{-1}$ and
$L_{8,21}$.
\end{enumerate}
\end{theorem}

\begin{proof}
We prove the assertion by direct analysis of the deformations of
the preceding algebras.
\begin{enumerate}

\item Let $L_{8,2}(\epsilon)=L_{8,2}+\epsilon\varphi$ be a linear
deformation of $L_{8,2}$. For any value of $\epsilon$ the deformed
commutator satisfies the Jacobi identity, thus defines a Lie
algebra. Computing the Killing metric tensor over the basis
$\left\{X_{1},..,X_{8}\right\}$, we obtain the matrix
\begin{equation*}
\kappa=\left(\begin{array}[c]{cccccccc}
 -3 & 0 & 0 & 0 & 0 & 0 & 0 & 0\\
 0 & -2 & 0 & 0 & 0 & 0 & 0 & \frac{3}{2}\epsilon\\
 0 & 0 & -3 & 0 & 0 & 0 & 0 & 0 \\
 0 & 0 & 0 & -6\epsilon & 0 & 0 & 0 & 0\\
 0 & 0 & 0 & 0 & -6\epsilon & 0 & 0 & 0\\
 0 & 0 & 0 & 0 & 0 & -6\epsilon & 0 & 0\\
0 & 0 & 0 & 0 & 0 & 0 & -6\epsilon & 0\\
0 & \frac{3}{2}\epsilon & 0 & 0 & 0 & 0 & 0 & -9\epsilon^{2}
\end{array}\right)
\end{equation*}
We have $\det(\kappa)= 2^{2}3^{8}7\epsilon^{6}\neq 0$ for
$\epsilon\neq 0$, and therefore the deformation is semisimple. To
identify to which real form $\frak{g}_{\epsilon}$ is isomorphic,
we compute the spectrum of $\kappa$ and obtain
\begin{equation}
{\rm Spec}(\kappa)=\left\{-3,3,(-6\epsilon)^{3},
-\frac{9}{2}\epsilon^{2}-1\pm
\frac{1}{2}\sqrt{\left(9\epsilon^{2}-\frac{3}{2}\right)+\frac{7}{4}}\right\}.
\end{equation}
Since $9\epsilon^{2}+2>
\sqrt{\left(9\epsilon^{2}-\frac{3}{2}\right)+\frac{7}{4}}$ for any
$\epsilon$, the two last roots of Spec$(\kappa)$ are always
negative, and the signature $\sigma$ of $\kappa$ is given by
\begin{equation*}
\sigma\left(  \kappa\right)  =\left\{
\begin{array}
[c]{rl}%
-8, & \varepsilon>0\\
0, & \varepsilon<0
\end{array}
\right.  .
\end{equation*}
For $\sigma=-8$ we obtain the compact Lie algebra $\frak{su}(3)$,
while for $\sigma=0$ we get the pseudo-unitary algebra
$\frak{su}(2,1)$ \cite{Co}. Finally, starting from the deformed
bracket, applying formula (\ref{Ko}) to the family of linear maps
defined by $f_{t}(X_{i})=X_{i}, (i=1,2,3)$,
$f_{t}(X_{i})=t^{-1}X_{i}, (i=4,...,7)$ and
$f_{t}(X_{8})=t^{-2}X_{8}$, we obtain the contraction of
$\frak{su}(2,1)$ onto $L_{8,2}$.

\item Let
$L_{8,4}^{0}(\epsilon_{1},\epsilon_{2})=L_{8,4}^{0}+\epsilon_{1}\varphi_{1}+\epsilon_{2}\varphi_{2}$
be a deformation. In this case the integrability condition is
$\epsilon_{1}\epsilon_{2}=0$. It is straightforward to verify that
the linear deformation
$L_{8,4}^{0}(0,\epsilon_{2})=L_{8,4}^{0}+\epsilon_{2}\varphi_{2}$
has a codimension one derived ideal, and cannot be
semisimple.\footnote{Actually this deformation leads to the Lie
algebra $L_{8,4}^{\epsilon_{2}}$. Since the latter algebra has no
invariants, it cannot be further deformed onto a semisimple
algebra.} Considering $L_{8,4}^{0}(\epsilon_{1},0)$ and computing
the spectrum of the Killing tensor $\kappa$, we obtain
\begin{equation*}
{\rm
Spec}(\kappa)=\left\{(-3)^{3},-4,(-6\epsilon_{1})^{4}\right\}.
\end{equation*}
Thus
\begin{equation*}
\sigma\left(  \kappa\right)  =\left\{
\begin{array}
[c]{rl}%
-8, & \epsilon_{1}>0\\
0, & \epsilon_{1}<0
\end{array}
\right.  ,
\end{equation*}
and we again obtain that
$L_{8,4}^{0}+\epsilon_{1}\varphi_{1}\simeq \frak{su}(3)$ if
$\epsilon_{1}>0$ and $L_{8,4}^{0}+\epsilon_{1}\varphi_{1}\simeq
\frak{su}(2,1)$ if $\epsilon_{1}<0$. Defining on
$L_{8,4}^{0}(\epsilon_{1},0)$ the linear maps
\begin{equation*}
f_{t}(X_{i})=X_{i}, (i=1,2,3,8);\; f_{t}(X_{i})=t^{-1}X_{i},
(i=4,...,7),
\end{equation*}
it follows that the contraction defined by them for $t\rightarrow
\infty$ is isomorphic to $L_{8,4}^{0}$, showing the invertibility
of the deformations.

\item For $L_{8,5}(\epsilon)=L_{8,5}+\epsilon\varphi$, the
spectrum of $\kappa$ is given by
\begin{equation*}
{\rm
Spec}(\kappa)=\left\{(-12)^{2},-8,-4\epsilon,-6\epsilon,(-24\epsilon)^{2},-72\epsilon\right\},
\end{equation*}
thus $\sigma(\kappa)=-8$ if $\epsilon>0$ and $\sigma(\kappa)=2$ if
$\epsilon<0$. This proves that $L_{8,5}(\epsilon)\simeq
\frak{su}(3)$ if $\epsilon>0$ and $L_{8,5}(\epsilon)\simeq
\frak{sl}(3,\mathbb{R})$ otherwise. The contraction reversing the
deformations are defined by the transformations
\begin{equation*}
f_{t}(X_{i})=X_{i}, (i=1,2,3);\; f_{t}(X_{i})=t^{-1}X_{i},
(i=4,...,8).
\end{equation*}

\item  For $L_{8,13}^{\varepsilon}$ we consider the deformations
$L_{8,13}^{\varepsilon}\left(
\mu\right)=L_{8,13}^{\epsilon}+\mu\varphi $, where
$\epsilon=\pm1$. For any nonzero $\mu$, we obtain
\begin{equation*}
{\rm Spec}\left(  \kappa\right)  =\left\{  -6,6,12,\left(
-12\mu\right) ^{2},\left( 12\mu\right)
^{2},-36\epsilon\mu^{2}\right\} .
\end{equation*}
The signature is $\sigma\left( \kappa\right)  =0$ for $\epsilon=1$
and $2$ for $\epsilon=-1$, proving that $L_{8,13}^{1}\left(
\mu\right) $ is isomorphic to $\frak{su}\left(  2,1\right)  $ and
$L_{8,13}^{-1}(\mu)$ is isomorphic to $\frak{sl}(3,\mathbb{R})$.
In both cases, the contractions are obtained from the changes of
basis in $L_{8,13}^{\epsilon}(\mu)$ defined by
\begin{equation*}
\fl f_{t}(X_{i})=X_{i}, (i=1,2,3);\; f_{t}(X_{i})=t^{-1}X_{i},
(i=4,...,7);\; f_{t}(X_{8})=t^{-2}X_{8}.
\end{equation*}

\item  Let $L_{8,15}\left(  \varepsilon\right)
=L_{8,15}+\varepsilon\varphi$ be a formal deformation. The
computation of the Killing form gives $\det\left(  \kappa\right)
=2^{14}3^{8}\varepsilon^{5}\neq0$ for nonzero $\varepsilon$, and
the spectrum is given by
\begin{equation*}
{\rm Spec}\left(  \kappa\right)  =\left\{  -6,6,12,\left(
-12\varepsilon\right) ^{3},\left(  12\varepsilon\right)
^{2}\right\}  ,
\end{equation*}
thus $\sigma\left(  \kappa\right)  =0$ for positive $\varepsilon$
and $\sigma\left(  \kappa\right)  =2$ for $\varepsilon<0$. We
obtain the deformations
\begin{equation*}
L_{8,15}+\varepsilon\varphi\simeq\left\{
\begin{array}
[c]{cc}%
\frak{su}\left(  2,1\right)  , & \varepsilon>0\\
\frak{sl}\left(  3,\mathbb{R}\right)  , & \varepsilon<0
\end{array}
\right.  .
\end{equation*}
The deformations are reversed considering the linear maps
\begin{equation*}
\fl f_{t}(X_{i})=X_{i}, (i=1,2,3);\; f_{t}(X_{i})=t^{-3}X_{i},
(i=4,5);\; f_{t}(X_{i})=t^{-1}X_{i}, (i=6,7);\;
f_{t}(X_{8})=t^{-2}X_{8}.
\end{equation*}

\item  The deformations of $L_{8,17}^{-1}$ were already considered
in \cite{C63}. Here we only note that
$L_{8,17}^{-1}+\varepsilon_{1}\varphi_{1}$ is semisimple, and that
the spectrum of $\kappa$ equals $\sigma\left( \kappa\right)  =2$
for any nonzero values of $\varepsilon_{1}$. We therefore obtain
the  deformation $L_{8,17}^{-1}+\varepsilon_{1}\varphi
_{1}\simeq\frak{sl}\left(  3,\mathbb{R}\right)  $, the
corresponding contraction being determined by
\begin{equation*}
f_{t}(X_{i})=X_{i}, (i=1,2,3,8);\; f_{t}(X_{i})=t^{-1}X_{i},
(i=4,...,7).
\end{equation*}

\item  The integrability condition for the deformation $L_{8,18}%
^{0}+\varepsilon_{1}\varphi_{1}+\varepsilon_{2}\varphi_{2}$ is
$\varepsilon _{1}\varepsilon_{2}=0$. Since the deformation
$L_{8,18}^{0}+\varepsilon _{2}\varphi_{2}$ leads to a Lie algebra
having no invariants, it cannot
provide any contraction of a semisimple algebra. Considering $L_{8,18}%
^{0}+\varepsilon_{1}\varphi_{1}$ and computing the Killing tensor,
we obtain $\det\left(  \kappa\right)
=2^{14}3^{7}\varepsilon_{1}^{4}\neq0$ and the spectrum
\begin{equation*}
{\rm Spec}\left(  \kappa\right)  =\left\{  -6,-4,6,12,\left(
-12\varepsilon _{1}\right)  ^{2},\left(  12\varepsilon_{1}\right)
^{2}\right\}  ,
\end{equation*}
and in any case $\sigma\left(  \kappa\right)  =0$, showing that $L_{8,18}%
^{0}+\varepsilon_{1}\varphi_{1}\simeq\frak{su}\left(  2,1\right)
$. To obtain the contraction, we consider on the deformations the
transformations
\begin{equation*}
f_{t}(X_{i})=X_{i}, (i=1,2,3,8);\; f_{t}(X_{i})=t^{-1}X_{i},
(i=4,...,7).
\end{equation*}

\item  Since the Lie algebra $L_{8,21}$ has the Levi decomposition
$\frak{sl}\left(  2,\mathbb{R}\right)
\overrightarrow{\oplus}_{D_{2}}5L_{1}$ and $\dim H^{2}\left(
L_{8,21},L_{8,21}\right)  \neq0$, it follows from \cite{Ri} that
$L_{8,21}$ is the contraction of a semisimple Lie algebra.
Considering $L_{8,21}+\varepsilon\varphi$, the spectrum of
$\kappa$ is given by
\begin{equation*}
{\rm Spec}\left(  \kappa\right)  =\left\{
-24,24,48,-48\varepsilon,-12\varepsilon
,12\varepsilon,48\varepsilon,72\varepsilon\right\}
\end{equation*}
and $\sigma\left(  \kappa\right)  =2$ for
$\varepsilon>0,\;\sigma\left( \kappa\right)  =0$ for
$\varepsilon<0$. We thus obtain that
\begin{equation*}
L_{8,21}+\varepsilon\varphi\simeq\left\{
\begin{array}
[c]{cc}%
\frak{su}\left(  2,1\right)  , & \varepsilon<0\\
\frak{sl}\left(  3,\mathbb{R}\right)  , & \varepsilon>0
\end{array}
\right.  .
\end{equation*}
In both cases, the contractions follows at once from the linear
maps
\begin{equation*}
f_{t}(X_{i})=X_{i}, (i=1,2,3);\; f_{t}(X_{i})=t^{-1}X_{i},
(i=4,...,8).
\end{equation*}
\end{enumerate}
To finish the proof, we still have to see that $L_{8,14}$ is also
a contraction of the non-compact simple algebras $\frak{su}(2,1)$
and $\frak{sl}(3,\mathbb{R})$. The Lie algebra $L_{8,14}$ has
three cocycle classes. Considering a
formal deformation $L_{8,14}\left(  \varepsilon_{1},\varepsilon_{2}%
,\varepsilon_{3}\right)
=L_{8,14}+\varepsilon_{1}\varphi_{1}+\varepsilon
_{2}\varphi_{2}+\varepsilon_{3}\varphi_{3}$, we obtain the
integrability conditions
\begin{equation*}
\varepsilon_{1}\varepsilon_{3}=\varepsilon_{2}\varepsilon_{3}=0.
\end{equation*}
A simple calculation shows that $L_{8,14}\left(  \varepsilon_{1}%
,\varepsilon_{2},0\right)  $ satisfies $\mathcal{N}\left(
L_{8,14}\left( \varepsilon_{1},\varepsilon_{2},0\right)  \right)
=0$ whenever $\varepsilon _{1}+\varepsilon_{2}\neq0$, and
therefore no simple algebra can be reached. If
$\varepsilon_{2}=-\varepsilon_{1}$, the deformation is isomorphic
to $L_{8,17}^{-1}$. Consider in $L_{8,14}\left(
\varepsilon_{1},-\varepsilon _{1},0\right)  \simeq L_{8,17}^{-1}$
the following change of basis
\begin{equation*}
\fl X_{4}^{\prime}=2X_{4},\;X_{5}^{\prime}=2X_{5},\;X_{6}^{\prime}=X_{4}%
+X_{6},\;X_{7}^{\prime}=X_{5}+X_{7},
\end{equation*}
where the remaining generators are not changed. This change is
easily seen to preserve the action of $\frak{sl}\left(
2,\mathbb{R}\right)  $ over the
radical. The only modified brackets are%
\begin{equation*}
\left[  X_{4}^{\prime},X_{8}^{\prime}\right]
=X_{4}^{\prime},\;\left[ X_{5}^{\prime},X_{8}^{\prime}\right]
=X_{5}^{\prime},\;\left[  X_{6}^{\prime },X_{8}^{\prime}\right]
=X_{4}^{\prime}-X_{6}^{\prime},\;\left[
X_{7}^{\prime},X_{8}^{\prime}\right]
=X_{5}^{\prime}-X_{7}^{\prime}.
\end{equation*}
If we now define
\begin{equation*}
\fl f_{t}\left(  X_{i}^{\prime}\right)
=X_{i}^{\prime},\;i=1,2,3;\;f_{t}\left( X_{i}^{\prime}\right)
=\frac{1}{t^{3}}X_{i}^{\prime},\;i=4,5;\;f_{t}\left(
X_{i}^{\prime}\right)
=\frac{1}{t}X_{i}^{\prime},\;i=6,7;\;f_{t}\left(
X_{8}^{\prime}\right)  =\frac{1}{t^{2}}X_{8}^{\prime},
\end{equation*}
then, over the new transformed basis $\left\{  X_{i}^{\prime\prime}%
=f_{t}\left(  X_{i}^{\prime}\right)  \right\}  $, the preceding
brackets are expressed by
\begin{equation*}
\fl \left[  X_{4}^{\prime\prime},X_{8}^{\prime\prime}\right]  =\frac{1}{t^{2}%
}X_{4}^{\prime\prime},\;\left[
X_{5}^{\prime\prime},X_{8}^{\prime\prime }\right]
=\frac{1}{t^{2}}X_{5}^{\prime\prime},\;\left[  X_{6}^{\prime\prime
},X_{8}^{\prime\prime}\right]  =X_{4}^{\prime\prime}-\frac{1}{t^{2}}%
X_{6}^{\prime\prime},\;\left[
X_{7}^{\prime\prime},X_{8}^{\prime\prime }\right]
=X_{5}^{\prime\prime}-\frac{1}{t^{2}}X_{7}^{\prime}.
\end{equation*}
For $t\rightarrow\infty$ we recover the brackets of $L_{8,14}$,
which shows that the deformation  $L_{8,14}\left(
\varepsilon_{1},-\varepsilon _{1},0\right)  $ is invertible. The
contraction
\begin{equation}
\frak{sl}\left(  3,\mathbb{R}\right)  \longrightarrow L_{8,14}
\label{K1}
\end{equation}
follows from transitivity of contractions \cite{We}. \newline
Finally, if we consider the deformation $L_{8,14}\left(
0,0,\varepsilon_{3}\right)  $ and the change of basis
$X_{8}^{\prime}=\epsilon_{3}X_{8}$, we immediately obtain that
this deformation is isomorphic to $L_{8,15}$. The corresponding
contraction of $L_{8,14}\left( 0,0,\varepsilon_{3}\right)  $ onto
$L_{8,14}$ is given by the family of transformations
\begin{equation*}
f_{t}(X_{i})=X_{i},\; (i=1,2,3,8);\quad
f_{t}(X_{i})=t^{-1}X_{i},\; (i=4,5,6,7).
\end{equation*}
Again, by transitivity of contractions, we obtain the previous
contraction (\ref{K1}) and also
\begin{equation}
\frak{su}\left(  2,1\right)  \longrightarrow L_{8,14}.
\end{equation}
\end{proof}

It remains to obtain the contractions on decomposable Lie algebras
$\frak{g}=\frak{g}_{1}\oplus\frak{g}_{2}$ with nonzero Levi part.
By corollary 1, none of the ideals $\frak{g}_{i}$ can be
semisimple. We can therefore assume that $\frak{g}_{1}$ has the
form
$\frak{g}_{1}=\frak{s}^{\prime}\overrightarrow{\oplus}_{R}\frak{r}$,
where $\frak{s}^{\prime}$ is simple of rank one and $5\leq
\dim\frak{g}_{1}\leq 7$. Then $\frak{g}$ can be rewritten as
\begin{equation*}
\frak{g}=\frak{s}^{\prime}\overrightarrow{\oplus}_{R\oplus
kD_{0}}\left(\frak{r}\oplus \frak{g}_{2}\right),
\end{equation*}
where $k=1,2,3$. Further, the embedding
$\frak{s}^{\prime}\hookrightarrow \frak{s}$ induces the branching
rule
\begin{equation}
ad(\frak{s})|_{\frak{s}^{\prime}}=ad(\frak{s}^{\prime})\oplus
R\oplus kD_{0}.
\end{equation}
By proposition 6, the multiplicity of the trivial representation
$D_{0}$ is at most one, from which it follows at once that the
only possibilities are $k=1, R=R_{4}^{II}$ if
$\frak{s}^{\prime}\simeq \frak{so}(3)$ and $k=1,
R=2D_{\frac{1}{2}}$ if $\frak{s}^{\prime}\simeq
\frak{sl}(2,\mathbb{R})$. This means that the simple algebras
$\frak{su}(3),\; \frak{su}(2,1)$ and $\frak{sl}(3,\mathbb{R})$ can
only contract onto the algebras $L_{7,2}\oplus L_{1}$ and
$L_{7,7}\oplus L_{1}$.

\begin{proposition}
The Lie algebra $L_{7,2}\oplus L_{1}$ is a contraction of
$\frak{su}(3)$ and $\frak{su}(2,1)$, while $L_{7,7}\oplus L_{1}$
is a contraction of $\frak{su}(2,1)$ and
$\frak{sl}(3,\mathbb{R})$.
\end{proposition}

\begin{proof}
We prove the assertion for $L_{7,2}\oplus L_{1}$, the reasoning
being similar for the remaining case. With some effort it can be
proved that $\dim H^{2}(L_{7,2}\oplus L_{1},L_{7,2}\oplus
L_{1})=7$. Considering the cocycles defined by
\begin{eqnarray*}
\fl \varphi_{1}(X_{4},X_{5})=X_{8},\;
\varphi_{1}(X_{6},X_{7})=-X_{8}\\
\fl \varphi_{2}(X_{4},X_{8})=-X_{6},\;
\varphi_{2}(X_{5},X_{8})=-X_{7},\;
\varphi_{2}(X_{6},X_{8})=X_{4},\; \varphi_{2}(X_{7},X_{8})=X_{5}
\end{eqnarray*}
and the corresponding linear deformations
$L(\epsilon_{1})=(L_{7,2}\oplus L_{1})+\epsilon_{1}\varphi_{1}$,
$L(\epsilon_{2})=(L_{7,2}\oplus L_{1})+\epsilon_{2}\varphi_{2}$,
it is not difficult to verify, using Table A1, that following
isomorphisms hold:
\begin{equation}
L(\epsilon_{1})\simeq L_{8,2},\quad L(\epsilon_{2})\simeq
L_{8,4}^{0}.
\end{equation}
Using the structure constants of Table A1, we define the changes
of basis
\begin{equation*}
f_{1,t}(X_{i})=X_{i},\;(i=1,2,3,8);\quad
f_{1,t}(X_{i})=t^{-1}X_{i},\;(i=4,5,6,7)
\end{equation*}
on $L(\epsilon_{1})$ and
\begin{equation*}
f_{2,t}(X_{i})=X_{i},\;(i=1,...,7);\quad
f_{2,t}(X_{8})=t^{-1}X_{8}
\end{equation*}
on $L(\epsilon_{2})$. A simple computation shows that the brackets
\begin{equation}
\left[X,Y\right]^{\prime}=\lim_{t\rightarrow
\infty}f_{k,t}^{-1}\left[f_{k,t}(X),f_{k,t}(Y)\right],\quad k=1,2
\end{equation}
are exactly those of $L_{7,2}\oplus L_{1}$. Since $\frak{su}(3)$
and $\frak{su}(2,1)$ both contract onto $L_{8,2}$ and
$L_{8,4}^{0}$, the result follows from transitivity of
contractions.

For $L_{7,7}\oplus L_{1}$ we also find that $\dim
H^{2}(L_{7,7}\oplus L_{1},L_{7,7}\oplus L_{1})=7$.
Considering the nontrivial cocycles%
\begin{eqnarray*}
\fl \varphi_{1}\left(  X_{4},X_{5}\right)  =X_{8};\quad
\varphi_{2}\left( X_{6},X_{7}\right)  =X_{8};\quad
\varphi_{3}\left( X_{6},X_{8}\right) =X_{4};\quad
\varphi_{3}\left(  X_{7},X_{8}\right)  =X_{5};\\
\fl \varphi_{4}\left(  X_{4},X_{8}\right)  =X_{4};\quad
\varphi_{4}\left( X_{5},X_{8}\right)  =X_{5};\quad
\varphi_{5}\left( X_{6},X_{8}\right) =X_{6};\quad
\varphi_{5}\left(  X_{7},X_{8}\right)  =X_{7};\\
\fl \varphi_{6}\left(  X_{4},X_{8}\right)  =X_{6};\quad
\varphi_{6}\left( X_{5},X_{8}\right)  =X_{7}.
\end{eqnarray*}
we obtain the following isomorphisms

\begin{enumerate}
\item $\frak{g}_{1}=\left(  L_{7,7}\oplus L_{1}\right)
+\varepsilon
_{1}\varphi_{1}+\varepsilon_{2}\varphi_{2}\simeq L_{13}^{\varepsilon}%
,$\newline $\left[  f_{1,t}\left(  X_{i}\right)  =X_{i},\;\left(
i=1,2,3,8\right)  ;\;f_{1,t}\left(  X_{i}\right)
=t^{-1}X_{i},\;\left( i=4,...,7\right)  \right]  \;$

\item $\frak{g}_{2}=\left(  L_{7,7}\oplus L_{1}\right)
+\varepsilon \varphi_{3}\simeq L_{8,14},\;$\newline $\left[
f_{2,t}\left(  X_{i}\right) =X_{i},\;\left(  i=1,...,7\right)
;\;f_{2,t}\left(  X_{8}\right) =t^{-1}X_{8}\right]  $

\item $\frak{g}_{3}=\left(  L_{7,7}\oplus L_{1}\right)
+\varepsilon _{1}\varphi_{2}+\varepsilon_{2}\varphi_{3}\simeq
L_{8,15},\;$\newline $\left[
f_{3,t}\left(  X_{i}\right)  =X_{i},\;\left(  i=1,...,5\right)  ;\;f_{3,t}%
\left(  X_{i}\right)  =t^{-1}X_{i},\;\left(  i=6,7,8\right)
\right]  $

\item $\frak{g}_{4}=\left(  L_{7,7}\oplus L_{1}\right)
+\varepsilon\left( \varphi_{4}-\varphi_{5}\right)  \simeq
L_{8,17}^{-1},\;$\newline $\left[
f_{4,t}\left(  X_{i}\right)  =X_{i},\;\left(  i=1,...,7\right)  ;\;f_{4,t}%
\left(  X_{8}\right)  =t^{-1}X_{8}\right]  $

\item $\frak{g}_{5}=\left(  L_{7,7}\oplus L_{1}\right)
+\varepsilon\left( \varphi_{5}-\varphi_{6}\right)  \simeq
L_{8,18}^{0},\;$\newline $\left[
f_{5,t}\left(  X_{i}\right)  =X_{i},\;\left(  i=1,...,7\right)  ;\;f_{5,t}%
\left(  X_{8}\right)  =t^{-1}X_{8}\right]  .$
\end{enumerate}
The linear maps in square brackets are defined over
$\frak{g}_{k}\;\left(
k=1,..,5\right)  $, and the limit%
\begin{equation*}
\left[  X,Y\right]
^{\prime}=\lim_{t\rightarrow\infty}f_{k,t}^{-1}\left[
f_{k,t}\left(  X\right)  ,f_{k,t}\left(  Y\right)  \right]
,\;X,Y\in \frak{g}_{k},\;k=1,..,5
\end{equation*}
exists for any pair of generators, thus define a contraction. It
can be easily verified that $\left[  X,Y\right]  ^{\prime}$
reproduces the brackets of $L_{7,7}\oplus L_{1}$. Again, the
contractions from the simple algebras $\frak{su}\left(  2,1\right)
$ and $\frak{sl}\left(  3,\mathbb{R}\right)  $ follow by
transitivity of contractions.
\end{proof}

\smallskip
In Figure 1 we display all the non-solvable contractions of
$\frak{su}(3)$, $\frak{su}(2,1)$ and $\frak{sl}(3,\mathbb{R})$
obtained in the previous results.

\begin{figure}
\caption{Non-solvable contractions of simple Lie algebras in dimension $8$.}%
\begin{diagram}
& & & & \frak{su}(3) & & & &\\
& & & \ldTo &   \dTo  & \rdTo(2,3) \rdTo & & &\\
& & L_{8,5} &  & L_{7,2}\oplus L_{1} & \lTo & L_{8,2} & &\\
& \ruTo &   &   &   & \luTo(2,1)  &  L_{8,4}^{0} &    \luTo &  \\
\frak{sl}(3,\mathbb{R}) & \rTo  & & & L_{8,15} &   &   &\lTo \luTo(2,1) & \frak{su}(2,1)\\
\dTo  & \rdTo \rdTo(2,4) \rdTo(4,2) \rdTo(5,4) & & & \dTo   & \rdTo(1,4) & & \ldTo(4,2)  \ldTo(6,4) \ldTo \ldTo(3,4) & \dTo \\
L_{8,13}^{-1} &  &  L_{8,17}^{-1} & \rTo & L_{8,14} & & L_{8,13}^{1} &  & L_{8,18}^{0}\\
&\rdTo(5,2)  &   & \rdTo(3,2) &   & \rdTo(1,2) \ldTo(1,2) &   & \ldTo(3,2) &\\
& & L_{8,21} & & & L_{7,7}\oplus L_{1} & & &\\
& & & & & & & &\\
\end{diagram}
\end{figure}

\section*{Concluding remarks}

We have determined all the non-solvable contractions of semisimple
Lie algebras up to dimension 8. Using the stability theorem of
Page and Richardson, we have obtained a first reduction of the
problem, and seen that the existence of contractions is determined
by the Levi decomposition of the target algebras. Moreover, it has
been pointed out that the embeddings of semisimple algebras in
other semisimple Lie algebras and the  associated branching
 rules are essential for the study of deformations and
 contractions in the non-solvable case, and show that decomposable
 and indecomposable algebras must be considered separately.
The next natural step of our analysis is to extend it to more
ample classes of target algebras, in order to
 determine the contractions of semisimple algebras onto
solvable Lie algebras. However, this problem can, in principle, be
solved only up to dimension six, since no classification of seven
dimensional solvable algebras is known. Further, the problem is
technically a formidable task, not only because of the large
number of isomorphism classes, but also because solvable algebras
can depend on many parameters, and therefore the deformations must
be analyzed for all possibilities of these parameters separately.
The recent work \cite{Po} shows the difficulties that appear even
in dimension four. Another possibility that is conceivable is to
compute all deformations and contractions among Lie algebras with
nontrivial Levi decomposition. In this sense, the only case having
been analyzed corresponds to the classical kinematical algebras
\cite{Ba}, corresponding to the representation of $\frak{so}(3)$
related to space isotropy. In the general problem, by the
Page-Richardson theorem, this task is reduced to analyze the
problem for Lie algebras having the same describing representation
$R$. While our analysis covers the dimensions six and seven, in
dimension 8 there are various parameterized families, and the
exact obtainment of all possible deformations (and contractions)
requires a large amount of special cases. Here the existence of
many non-invertible deformations makes the analysis quite
complicated. Work in this direction is actually in progress.

\medskip
Among the applications of the results obtained here, we enumerate
the missing label problem and the spontaneous symmetry breaking.
Especially for the case of semisimple algebras, the knowledge of
the contractions preserving some semisimple subalgebra is of
interest in many situations. A special case is given by
inhomogeneous Lie algebras \cite{He,C43}. However, other types of
semidirect products are relevant for many problems, such as the
Galilei, Schr\"odinger or the Poincar\'e-Heisenberg algebras, and
their deformations often provide additional information on the
states or the configuration of a system and their invariants
\cite{Ch,C46,Ca,Gui}. In the case of the missing label problem,
the contractions can be used to determine additional operators
that commute with the subalgebra \cite{C46}. Finally, the obtained
contractions could also be of interest in establishing relations
among completely integrable systems defined on contractions of
semisimple Lie algebras \cite{Bol}.

\section*{Acknowledgements}
The author wishes to express his gratitude to the referees for
useful suggestions that helped to improve the manuscript, as well
as for reference \cite{Gui}. This work was partially supported
 by the research project MTM2006-09152 of the Ministerio de Educaci\'on y
 Ciencia. The figures of the paper were prepared with the help of the LaTex
 package \textit{commutative diagrams} by Paul Taylor.

\section*{Appendix}
\appendix
\setcounter{section}{1}

In this appendix we give the structure constants of Lie algebras
in dimension $n\leq 8$ having a nontrivial Levi decomposition,
following the notation of the original classification \cite{Tu}.
The brackets are expressed by
$\left[X_{i},X_{j}\right]=C_{ij}^{k}X_{k}$ over the ordered basis
$\left\{X_{1},..,X_{n}\right\}$ of $\frak{g}$.

\begin{table}
\begin{indented}\item[]
\caption{Structure constants for indecomposable Lie algebras with
nontrivial Levi decomposition in dimension $n\leq 8$ after
\cite{Tu}.}
\begin{tabular}{@{}ll}
Algebra & \textrm{Structure constants} \\ \hline\mr
$L_{5,1}$ & $C_{12}^{2}=2,C_{13}^{3}=-2,%
\;C_{23}^{1}=1,C_{14}^{4}=1,C_{25}^{4}=1,C_{34}^{5}=1,C_{15}^{5}=-1.$ \\
$L_{6,1}$ & $%
C_{23}^{1}=1,C_{12}^{3}=1,C_{13}^{2}=-1,C_{15}^{6}=1,C_{16}^{5}=-1,C_{24}^{6}=-1,
C_{26}^{4}=1,C_{34}^{5}=1,C_{35}^{4}=-1.
$ \\
$L_{6,2}$ & $%
C_{12}^{2}=2,C_{13}^{3}=-2,C_{23}^{1}=1,C_{14}^{4}=1,C_{25}^{4}=1,C_{34}^{5}=1,C_{15}^{5}=-1,C_{45}^{6}=1.
$ \\
$L_{6,3}$ & $%
C_{12}^{2}=2,C_{13}^{3}=-2,C_{23}^{1}=1,C_{14}^{4}=1,C_{25}^{4}=1,C_{34}^{5}=1,C_{15}^{5}=-1,C_{j6}^{j}=1,\;\left(
j=4,5\right) .
$ \\
$L_{6,4}$ & $%
C_{12}^{2}=2,C_{13}^{3}=-2,C_{23}^{1}=1,C_{14}^{4}=2,C_{16}^{6}=-2,C_{25}^{4}=2,C_{26}^{5}=1,C_{34}^{5}=1,C_{35}^{6}=2.
$ \\
$L_{7,1}$ & $C_{23}^{1}=1,C_{12}^{3}=1,C_{13}^{2}=-1,%
\;C_{15}^{6}=1,C_{16}^{5}=-1,C_{24}^{6}=-1,C_{26}^{4}=1,C_{34}^{5}=1,C_{35}^{4}=-1,
$ \\
& $C_{j7}^{j}=1\;\left( 4\leq j\leq 6\right) .$ \\
$L_{7,2}$ & $C_{23}^{1}=1,C_{12}^{3}=1,C_{13}^{2}=-1,\;C_{14}^{7}=\frac{1}{2}%
,C_{15}^{6}=\frac{1}{2},C_{16}^{5}=-\frac{1}{2},C_{17}^{4}=-\frac{1}{2}%
,C_{24}^{5}=\frac{1}{2},$ \\
& $C_{25}^{4}=-\frac{1}{2},C_{26}^{7}=\frac{1}{2},C_{27}^{6}=-\frac{1}{2}%
,C_{34}^{6}=\frac{1}{2},C_{35}^{7}=-\frac{1}{2},C_{36}^{4}=-\frac{1}{2}%
,C_{37}^{5}=\frac{1}{2}.$ \\
$L_{7,3}$ & $%
C_{12}^{2}=2,C_{13}^{3}=-2,C_{23}^{1}=1,C_{14}^{4}=1,C_{15}^{5}=-1,C_{25}^{4}=1,\
C_{34}^{5}=1,C_{47}^{4}=1,C_{57}^{5}=1,
$ \\
& $C_{67}^{6}=p\;\left( p\neq 0\right) .$ \\
$L_{7,4}$ & $%
C_{12}^{2}=2,C_{13}^{3}=-2,C_{23}^{1}=1,C_{14}^{4}=1,C_{15}^{5}=-1,C_{25}^{4}=1,C_{34}^{5}=1,C_{45}^{6}=1,C_{47}^{4}=1,
$ \\
& $C_{57}^{5}=1,C_{67}^{6}=2.$ \\
$L_{7,5}$ & $%
C_{12}^{2}=2,C_{13}^{3}=-2,C_{23}^{1}=1,C_{14}^{4}=2,C_{16}^{6}=-2,C_{25}^{4}=2,C_{26}^{5}=1,\;C_{34}^{4}=1,C_{35}^{5}=2,
$ \\
& $C_{j7}^{j}=1\;\left( j=,4,5,6\right) .$ \\
$L_{7,6}$ & $C_{12}^{2}=2,C_{13}^{3}=-2,C_{23}^{1}=1,%
\;C_{14}^{4}=3,C_{15}^{5}=1,C_{16}^{6}=-1,C_{17}^{7}=-3,C_{25}^{4}=3,%
\;C_{26}^{5}=2,$ \\
& $C_{27}^{6}=1,C_{34}^{5}=1,C_{35}^{6}=2,C_{36}^{7}=3.$ \\
$L_{7,7}$ & $C_{12}^{2}=2,C_{13}^{3}=-2,C_{23}^{1}=1,%
\;C_{14}^{4}=1,C_{15}^{5}=-1,%
\;C_{25}^{4}=1,C_{27}^{6}=1,C_{34}^{5}=1,C_{16}^{6}=1,$ \\
& $C_{17}^{7}=-1,C_{36}^{7}=1.$ \\
$L_{8,1}$ & $C_{23}^{1}=1,C_{12}^{3}=1,C_{13}^{2}=-1,%
\;C_{15}^{6}=1,C_{16}^{5}=-1,C_{24}^{6}=-1,C_{26}^{4}=1,C_{34}^{5}=1,C_{35}^{4}=-1,
$ \\
& $C_{j8}^{j}=1\;\left( 4\leq j\leq 6\right) ,C_{78}^{7}=p.$ \\
$L_{8,2}$ & $C_{23}^{1}=1,C_{12}^{3}=1,C_{13}^{2}=-1,\;C_{14}^{7}=\frac{1}{2}%
,C_{15}^{6}=\frac{1}{2},C_{16}^{5}=-\frac{1}{2},C_{17}^{4}=-\frac{1}{2}%
,C_{24}^{5}=\frac{1}{2},C_{25}^{4}=-\frac{1}{2},$ \\
& $C_{26}^{7}=\frac{1}{2},C_{27}^{6}=-\frac{1}{2},C_{34}^{6}=\frac{1}{2}%
,C_{35}^{7}=-\frac{1}{2},C_{36}^{4}=-\frac{1}{2},C_{37}^{5}=\frac{1}{2}%
,C_{45}^{8}=1,C_{67}^{8}=-1.$ \\
$L_{8,3}$ & $C_{23}^{1}=1,C_{12}^{3}=1,C_{13}^{2}=-1,\;C_{14}^{7}=\frac{1}{2}%
,C_{15}^{6}=\frac{1}{2},C_{16}^{5}=-\frac{1}{2},C_{17}^{4}=-\frac{1}{2}%
,C_{24}^{5}=\frac{1}{2},C_{25}^{4}=-\frac{1}{2},$ \\
& $C_{26}^{7}=\frac{1}{2},C_{27}^{6}=-\frac{1}{2},C_{34}^{6}=\frac{1}{2}%
,C_{35}^{7}=-\frac{1}{2},C_{36}^{4}=-\frac{1}{2},C_{37}^{5}=\frac{1}{2}%
,C_{48}^{4}=1,C_{58}^{5}=1,C_{68}^{6}=1,$ \\
& $C_{78}^{7}=1.$ \\
$L_{8,4}^{p}$ & $C_{23}^{1}=1,C_{12}^{3}=1,C_{13}^{2}=-1,\;C_{14}^{7}=\frac{1%
}{2},C_{15}^{6}=\frac{1}{2},C_{16}^{5}=-\frac{1}{2},C_{17}^{4}=-\frac{1}{2}%
,C_{24}^{5}=\frac{1}{2},C_{25}^{4}=-\frac{1}{2},$ \\
& $C_{26}^{7}=\frac{1}{2},C_{27}^{6}=-\frac{1}{2},C_{34}^{6}=\frac{1}{2},C_{35}^{7}=-\frac{1}{2}%
,C_{36}^{4}=-\frac{1}{2},C_{37}^{5}=\frac{1}{2}%
,C_{48}^{48}=p,C_{58}^{5}=p,C_{68}^{6}=p,$ \\
& $C_{78}^{7}=p,C_{48}^{6}=-1,C_{58}^{7}=-1,C_{68}^{4}=1,C_{78}^{5}=1.$ \\
$L_{8,5}$ & $C_{23}^{1}=1,C_{12}^{3}=1,C_{13}^{2}=-1,C_{14}^{7}=\frac{1}{2}%
,C_{15}^{6}=-\frac{1}{2},C_{16}^{5}=2,C_{16}^{8}=-1,C_{17}^{4}=-2,C_{18}^{6}=3,$ \\
& $C_{24}^{6}=\frac{1}{2},C_{25}^{7}=\frac{1}{2}%
,C_{26}^{4}=-2,C_{27}^{5}=-2,C_{27}^{8}=-1,C_{28}^{7}=3,C_{34}^{5}=2,C_{35}^{4}=-2,$ \\
& $C_{36}^{7}=1,C_{37}^{6}=-1.$ \\
$L_{8,6}$ & $%
C_{12}^{2}=2,C_{13}^{3}=-2,C_{23}^{1}=1,C_{14}^{4}=1,C_{15}^{5}=-1,C_{25}^{4}=1,C_{34}^{5}=1,C_{45}^{8}=1,
C_{67}^{8}=1.$ \\
$L_{8,7}^{p,q}$ & $%
C_{12}^{2}=2,C_{13}^{3}=-2,C_{23}^{1}=1,C_{14}^{4}=1,C_{15}^{5}=-1,C_{25}^{4}=1,C_{34}^{5}=1,C_{48}^{4}=1,
C_{58}^{5}=1,$ \\
$pq\neq 0$ & $C_{68}^{6}=p,C_{78}^{7}=q.$ \\
$L_{8,8}^{p}$ & $%
C_{12}^{2}=2,C_{13}^{3}=-2,C_{23}^{1}=1,C_{14}^{4}=1,C_{15}^{5}=-1,C_{25}^{4}=1,C_{34}^{5}=1,
C_{48}^{4}=1,C_{58}^{5}=1,$ \\
$p\neq 0$ & $C_{68}^{6}=p,C_{78}^{6}=1,C_{78}^{7}=p.$ \\
$L_{8.8}^{0}$ & $%
C_{12}^{2}=2,C_{13}^{3}=-2,C_{23}^{1}=1,C_{14}^{4}=1,C_{15}^{5}=-1,C_{25}^{4}=1,C_{34}^{5}=1,C_{48}^{4}=1,
C_{58}^{5}=1,$ \\
&  $C_{78}^{6}=1.$ \\
$L_{8,9}^{p,q}$ & $%
C_{12}^{2}=2,C_{13}^{3}=-2,C_{23}^{1}=1,C_{14}^{4}=1,C_{15}^{5}=-1,C_{25}^{4}=1,C_{34}^{5}=1,C_{48}^{4}=1,
C_{58}^{5}=1,$ \\
$q\neq 0$ & $C_{68}^{6}=p,C_{68}^{7}=-q,C_{78}^{6}=q,C_{78}^{7}=p.$ \\
$L_{8,10}^{p}$ & $%
C_{12}^{2}=2,C_{13}^{3}=-2,C_{23}^{1}=1,C_{14}^{4}=1,C_{15}^{5}=-1,C_{25}^{4}=1,C_{34}^{5}=1,C_{48}^{4}=1,
C_{58}^{5}=1,$ \\
&  $C_{68}^{6}=2,C_{78}^{7}=p,C_{45}^{6}=1.$\\
&  \\ \hline\br
\end{tabular}
\end{indented}
\end{table}

\begin{table}
\begin{indented}\item[]
\caption{Structure constants for indecomposable Lie algebras with
nontrivial Levi decomposition in dimension $n\leq 8$ after
\cite{Tu} (cont.).}
\begin{tabular}{@{}ll}
Algebra & Structure constants \\ \hline\mr
$L_{8,11}$ & $%
C_{12}^{2}=2,C_{13}^{3}=-2,C_{23}^{1}=1,C_{14}^{4}=1,C_{15}^{5}=-1,C_{25}^{4}=1,C_{34}^{5}=1,C_{48}^{4}=1,
C_{58}^{5}=1,$ \\
& $C_{68}^{6}=2,C_{78}^{6}=1,C_{78}^{7}=2,C_{45}^{6}=1.$ \\
$L_{8,12}^{p}$ & $%
C_{12}^{2}=2,C_{13}^{3}=-2,C_{23}^{1}=1,C_{14}^{4}=2,C_{16}^{6}=-2,C_{25}^{4}=2,C_{26}^{5}=1,C_{34}^{5}=1,
C_{35}^{6}=2,$ \\
& $C_{48}^{4}=1,C_{58}^{5}=1,C_{68}^{6}=1,C_{78}^{7}=p.$ \\
$L_{8,13}^{\varepsilon }$ & $%
C_{12}^{2}=2,C_{13}^{3}=-2,C_{23}^{1}=1,C_{14}^{4}=1,C_{15}^{5}=-1,C_{25}^{4}=1,C_{34}^{5}=1,C_{16}^{6}=1,
C_{17}^{7}=-1,$ \\
& $C_{27}^{6}=1,C_{36}^{7}=1,C_{45}^{8}=1,C_{67}^{8}=\varepsilon .$ \\
$L_{8,14}$ & $%
C_{12}^{2}=2,C_{13}^{3}=-2,C_{23}^{1}=1,C_{14}^{4}=1,C_{15}^{5}=-1,C_{25}^{4}=1,C_{34}^{5}=1,C_{16}^{6}=1,
C_{17}^{7}=-1,$ \\
& $C_{27}^{6}=1,C_{36}^{7}=1,C_{68}^{4}=1,C_{78}^{5}=1.$ \\
$L_{8,15}$ & $%
C_{12}^{2}=2,C_{13}^{3}=-2,C_{23}^{1}=1,C_{14}^{4}=1,C_{15}^{5}=-1,C_{25}^{4}=1,C_{34}^{5}=1,C_{16}^{6}=1,
C_{17}^{7}=-1,$ \\
& $C_{27}^{6}=1,C_{36}^{7}=1,C_{67}^{8}=1,C_{68}^{4}=1,C_{78}^{5}=1.$ \\
$L_{8,16}$ & $%
C_{12}^{2}=2,C_{13}^{3}=-2,C_{23}^{1}=1,C_{14}^{4}=1,C_{15}^{5}=-1,C_{25}^{4}=1,C_{34}^{5}=1,C_{16}^{6}=1,
C_{17}^{7}=-1,$ \\
&
$C_{27}^{6}=1,C_{36}^{7}=1,C_{48}^{4}=1,C_{58}^{5}=1,C_{68}^{4}=1,C_{68}^{6}=1,C_{78}^{5}=1,C_{78}^{7}=1.
$ \\
$L_{8,17}^{p}$ & $%
C_{12}^{2}=2,C_{13}^{3}=-2,C_{23}^{1}=1,C_{14}^{4}=1,C_{15}^{5}=-1,C_{25}^{4}=1,C_{34}^{5}=1,C_{16}^{6}=1,
C_{17}^{7}=-1,$ \\
$p\neq -1$ & $C_{27}^{6}=1,C_{36}^{7}=1,C_{48}^{4}=1,C_{58}^{5}=1,C_{68}^{6}=p,C_{78}^{7}=p.$ \\
$L_{8,17}^{-1}$ & $%
C_{12}^{2}=2,C_{13}^{3}=-2,C_{23}^{1}=1,C_{14}^{4}=1,C_{15}^{5}=-1,C_{25}^{4}=1,C_{34}^{5}=1,C_{16}^{6}=1,
C_{17}^{7}=-1,$ \\
& $C_{27}^{6}=1,C_{36}^{7}=1,C_{48}^{4}=1,C_{58}^{5}=1,C_{68}^{6}=-1,C_{78}^{7}=-1.$ \\
$L_{8,18}^{p}$ & $%
C_{12}^{2}=2,C_{13}^{3}=-2,C_{23}^{1}=1,C_{14}^{4}=1,C_{15}^{5}=-1,C_{25}^{4}=1,C_{34}^{5}=1,C_{16}^{6}=1,
C_{17}^{7}=-1,$ \\
$p\neq 0$ &
$C_{27}^{6}=1,C_{36}^{7}=1,C_{48}^{4}=p,C_{48}^{6}=-1,C_{58}^{5}=p,C_{58}^{7}=-1,C_{68}^{4}=1,C_{68}^{6}=p,
C_{78}^{5}=1$,\\& $C_{78}^{7}=p.$ \\
$L_{8,18}^{0}$ & $%
C_{12}^{2}=2,C_{13}^{3}=-2,C_{23}^{1}=1,C_{14}^{4}=1,C_{15}^{5}=-1,C_{25}^{4}=1,C_{34}^{5}=1,C_{16}^{6}=1,
C_{17}^{7}=-1,$ \\
& $C_{27}^{6}=1,C_{36}^{7}=1,C_{48}^{6}=-1,C_{58}^{7}=-1,C_{68}^{4}=1,C_{78}^{5}=1.$ \\
$L_{8,19}$ & $%
C_{12}^{2}=2,C_{13}^{3}=-2,C_{23}^{1}=1,C_{14}^{4}=3,C_{15}^{5}=1,C_{16}^{6}=-1,C_{17}^{7}=-3,C_{25}^{4}=3,
C_{26}^{5}=2,$ \\
& $C_{27}^{6}=1,C_{34}^{5}=1,C_{35}^{6}=2,C_{36}^{7}=3,C_{47}^{8}=1,C_{56}^{8}=-3.$ \\
$L_{8,20}$ & $%
C_{12}^{2}=2,C_{13}^{3}=-2,C_{23}^{1}=1,C_{14}^{4}=3,C_{15}^{5}=1,C_{16}^{6}=-1,C_{17}^{7}=-3,C_{25}^{4}=3,
C_{26}^{5}=2,$ \\
&
$C_{27}^{6}=1,C_{34}^{5}=1,C_{35}^{6}=2,C_{36}^{7}=3,C_{48}^{4}=1,C_{58}^{5}=1,C_{68}^{6}=1,C_{78}^{7}=1.
$ \\
$L_{8,21}$ & $%
C_{12}^{2}=2,C_{13}^{3}=-2,C_{23}^{1}=1,C_{14}^{4}=4,C_{15}^{5}=2,C_{17}^{7}=-2,C_{18}^{8}=-4,C_{25}^{4}=4,
C_{26}^{5}=3,$ \\
& $C_{27}^{6}=2,C_{28}^{7}=1,C_{34}^{5}=1,C_{35}^{6}=2,C_{36}^{7}=3,C_{37}^{8}=4.$ \\
$L_{8,22}$ & $%
C_{12}^{2}=2,C_{13}^{3}=-2,C_{23}^{1}=1,C_{14}^{4}=2,C_{16}^{6}=-2,C_{17}^{7}=1,C_{18}^{8}=-1,C_{25}^{4}=2,
C_{26}^{5}=1,$ \\
& $C_{28}^{7}=1,C_{34}^{5}=1,C_{35}^{6}=2,C_{37}^{8}=1.$ \\
\hline\br
\end{tabular}
\end{indented}
\end{table}

\end{document}